\documentclass[a4paper,
11pt,openright]{article} 
\usepackage[english]{babel} 
\usepackage[latin1]{inputenc} 
\usepackage{amsmath,amssymb,amsthm,amsfonts,amscd,authblk,dsfont,color,cancel,diagbox,enumerate,epsfig,eucal,fancyhdr,latexsym,lmodern,mathrsfs,mathtools,multicol,multirow,musicography,phaistos,skull,simpler-wick,subcaption,tikz-cd}
\usepackage[normalem]{ulem} 
\usepackage[multiuser]{fixme}
\FXRegisterAuthor{nd}{annd}{Nico} 
\FXRegisterAuthor{cvdv}{ancvdv}{Chris} 
\usetikzlibrary{shapes,snakes} 

\usepackage[
final=true,                    
plainpages=false,              
pdfpagelabels=true,            
pdfencoding=auto,              
unicode=true,                  
hypertexnames=true,            
naturalnames=true              
]{hyperref}
\usepackage[all,cmtip]{xy}
\usetikzlibrary{matrix,calc}
\usepackage[autostyle, english = american]{csquotes}
\MakeOuterQuote{"} 

\definecolor{NiColor}{RGB}{77,77,255}
\definecolor{NiColoRed}{RGB}{255,77,77}
\definecolor{NiCitation}{RGB}{0,181,26}

\newtheoremstyle{TheoremStyle}
{3pt}
{3pt}
{}
{}
{\sc}
{:}
{.5em}
{}
\makeatletter
\def\@endtheorem{\hfill$\lozenge$} 
\makeatother

\theoremstyle{TheoremStyle}
\newtheorem{theorem}{Theorem}

\newtheorem{proposition}[theorem]{Proposition}

\newtheorem{remark}[theorem]{Remark}

\title{DLR-KMS correspondence on lattice spin systems}
\author[a]{N. Drago\thanks{\href{mailto:nicolo.drago@unitn.it}{nicolo.drago@unitn.it}}}
\author[b]{C. J. F. van de Ven\thanks{\href{mailto:christiaan.vandeven@mathematik.uni-wuerzburg.de}{christiaan.vandeven@mathematik.uni-wuerzburg.de}}}
\affil[a]{Dipartimento di Matematica, Universit\`a di Trento and INFN-TIFPA and INdAM, Via Sommarive 14, I-38123 Povo, Italy}
\affil[b]{Julius Maximilian University of W\"urzburg, Department of Mathematics Chair of Mathematics X (Mathematical Physics), Emil-Fischer-Stra\ss e 31, 97074 W\"urzburg, Germany}

\def\bearray{\begin{eqnarray}}
\def\earray{\end{eqnarray}}
\def\beq{\begin{equation}}
\def\eeq{\end{equation}}

\def\b0{{\bf 0}}

\newsymbol\bt 1202             

\newsymbol\rest 1316         

\begin{document} 

\maketitle

\begin{abstract}
    The Dobrushin-Lanford-Ruelle (DLR) condition \cite{Dobrushin_1970,Lanford_Ruelle_1969} and the classical Kubo-Martin-Schwinger (KMS) condition \cite{Gallavotti_Verboven_1975} are considered in the context of classical lattice systems.
    In particular, we prove that these conditions are equivalent for the case of a lattice spin system with values in a compact symplectic manifold by showing that infinite volume Gibbs states are in bijection with KMS states.
\end{abstract}


\section{Introduction}

Identifying the proper notion of thermal equilibrium in an infinite system is of paramount importance in statistical mechanics \cite{Ruelle_1969}.
A common approach to this problem is to consider a finite size approximation of the system under consideration, where the notion of thermal equilibrium is captured by the renown Gibbs state \cite{Friedly_Velenik_2017}, \textit{cf.} Equation \eqref{Eq: Gibbs measure on Lambda with eta boundary conditions}.
A subsequent limit procedure, where the size of the system diverges, leads to states which are identified (a posteriori) as those describing thermal equilibrium in the infinite system.
While being a convenient and fruitful approach, this method leaves open the problem of identifying thermal equilibrium states directly on the infinite system.
Such possibility is essential because it leads to a proper notion of phase transitions, independently on the chosen order parameters \cite[\S 3]{Friedly_Velenik_2017}.

For lattice spin systems, the Dobrushin-Lanford-Ruelle (DLR) condition provides a satisfactory solution to the stated problem \cite{Dobrushin_1968_1,Dobrushin_1968_2,Dobrushin_1968_3,Dobrushin_1970,Lanford_Ruelle_1969,Ruelle_1967}.
Therein, thermal equilibrium states are identified with probability measures which reduce to the Gibbs state when conditioned on the complement of a finite region, see Section \ref{Sec: DLR approach on continuous lattice system} for the precise definition.
States abiding by the DLR condition are usually called infinite-volume Gibbs states.
The main ingredients of this approach are the lattice formulation of the system of interest together with the notion of Gibbs states on a local region of such lattice \cite{Georgii_2011}.

Another conceptually clear yet different in spirit condition to pinpoint thermal equilibrium states is the renown Kubo-Martin-Schwinger (KMS) condition.
The quantum version of this condition, usually formulated in the $C^*$-algebraic setting, identifies thermal equilibrium in terms of an analytic condition on the correlations of the state \cite{Bratteli_Robinson_1981_II, Haag_Hugenholtz_Winnik_1967}.
A classical version of this condition is at disposal and has been investigated in different contexts ranging from systems of infinitely many particles on $\mathbb{R}^d$ \cite{Aizenman_Gallavotti_Goldstein_Lebowitz_1976,Aizenman_Goldstein_Gruber_Lebowitz_Martin_1977,Gallavotti_Pulvirenti_1976,Gallavotti_Verboven_1975} to Poisson geometry \cite{Basart_Flato_Lichnerowicz_Sternheimer_1984,Bordermann_Romer_Waldmann_1998,Drago_Waldmann_2021}.
The classical KMS condition fixes the expectation value of the state of interest on Poisson brackets, \textit{cf.} Section \ref{Sec: Classical KMS condition on lattice systems} for the precise definition.
Thus, it requires a Poisson structure on the system under investigation.

At this juncture a natural question arises, namely if the DLR and the (classical) KMS conditions agree whenever the system under investigation is described in terms of a lattice structure endowed with a Poisson bracket.
A positive answer for systems of particles on $\mathbb{R}^d$ has been found in \cite{Aizenman_Goldstein_Gruber_Lebowitz_Martin_1977} and similar positive results have been obtained on the quantum side \cite[Thm. 6.2.18]{Bratteli_Robinson_1981_II}.

In this paper we deal with the DLR-KMS correspondence for the case of a lattice system whose configuration space $\Omega=M^{\mathbb{Z}^d}$ is the space of functions $\mathbb{Z}^d\to M$ with values in a compact symplectic manifold $M$.
This provides a natural setting where both the DLR and KMS approaches apply, leading to the natural question of whether the DLR-KMS correspondence holds true also in this setting.
The main result of this paper is the following theorem, which provides a positive answer to this question.
\begin{theorem}\label{Thm: equivalence of oo-volume Gibbs states and KMS states}
	Let $\mathcal{G}(\Phi)$ be the convex set of infinite volume Gibbs states on $\Omega=M^{\mathbb{Z}^d}$ ---\textit{cf.} Section \ref{Sec: Classical KMS condition on lattice systems} and Equation \eqref{Eq: oo-volume Gibbs state property}.
	Let $\mathcal{K}(\Phi)$ be the convex set of $X^\Phi$-KMS states defined as per Equation \eqref{Eq: KMS condition}, \textit{cf.} Section \ref{Sec: Classical KMS condition on lattice systems}.
	Then $\mathcal{G}(\Phi)=\mathcal{K}(\Phi)$.
\end{theorem}

Aside from its relevance in the discussion on the notion of thermal equilibrium, Theorem \ref{Thm: equivalence of oo-volume Gibbs states and KMS states} is relevant for its connection with strict deformation quantization (SDQ) \cite{Bayen_Flato_Fronsdal_Lichnerowicz_1978_I,Bayen_Flato_Fronsdal_Lichnerowicz_1978_II,Berezin_1975}.
The latter provides a mathematically sound description of the classical limit of quantum theories and it is particularly suitable for the analysis of the semi-classical limit of quantum states.
In view of the recent results in this framework \cite{Drago_vandeVen_2022,Moretti_vandeVen_2022,Murro-vandeVen_2022,vandeVen_2022,vandeVen_2022_b}, Theorem \ref{Thm: equivalence of oo-volume Gibbs states and KMS states} stands as a technical result which may lead to a deeper understanding of the quantum-to-classical limit of thermal equilibrium states.
We plan to address this question in a future investigation.

From a technical point of view this paper profits from the strategy applied in \cite{Aizenman_Goldstein_Gruber_Lebowitz_Martin_1977} together with the results obtained in \cite{Bordermann_Romer_Waldmann_1998, Drago_Waldmann_2021}.
As a matter of fact the proof of $\mathcal{G}(\Phi)\subseteq\mathcal{K}(\Phi)$ is a direct computation, \textit{cf.} Proposition \ref{Prop: oo-volume Gibbs states are KMS states}.
Instead, the converse inclusion requires a more refined analysis of the conditional probability $\varphi_\Lambda(\;|\eta)$ of a given KMS state $\varphi$ with respect to the complement $\Lambda^c$ of a finite region $\Lambda\Subset\mathbb{Z}^d$.
This is where the results of \cite{Bordermann_Romer_Waldmann_1998,Drago_Waldmann_2021} apply, ensuring that such conditional probabilities $\varphi_\Lambda(\;|\eta)$ coincide with the local Gibbs state $\varphi^\Phi_\Lambda(\;|\eta)$.

The paper is organized as follows:
Section \ref{Sec: DLR approach on continuous lattice system} is devoted to a brief introduction to the lattice spin system of interest together with the precise definition of the DLR condition.
Similarly, Section \ref{Sec: Classical KMS condition on lattice systems} considers the classical KMS condition in the same setting.
Finally Section \ref{Sec: Proof of the classical DLR-KMS correspondence} proves Theorem \ref{Thm: equivalence of oo-volume Gibbs states and KMS states}, \textit{cf.} Proposition \ref{Prop: oo-volume Gibbs states are KMS states}-\ref{Prop: KMS states are oo-volume Gibbs states}.

\section{DLR approach on continuous lattice system}
\label{Sec: DLR approach on continuous lattice system}

In this section we briefly recollect a few crucial results and definitions from continuous lattice systems, see \cite{Friedly_Velenik_2017}.
In what follows we will consider a configuration space over the lattice $\mathbb{Z}^d$ with values in a symplectic manifold $M$.
The differential structure of $M$ is needed for the formulation of the classical KMS condition, \textit{cf.} Section \ref{Sec: Classical KMS condition on lattice systems}.

Let $(M,\varsigma)$ be a compact, connected symplectic manifold with symplectic form $\varsigma$.
For technical convenience we shall assume that $M$ is metrizable and denote with $d_M$ a complete metric on $M$ whose induced topology coincides with the one of $M$.
We will denote by $\mu_M:=c_M\varsigma^m/m!$ the induced volume form on $M$, where $\dim M=2m$ while the normalization constant $c_M>0$ is chosen so that $\int_M\mu_M=1$.
For $f,g\in C^\infty(M)$ we denote by $\{f,g\}$ the induced Poisson bracket between $f$ and $g$ and set $X_f:=\{\;,f\}$ the Hamiltonian vector field associated with $f$.

The configuration space of interest is $\Omega:=M^{\mathbb{Z}^d}$ which is compact in the product topology.
Moreover, $\Omega$ is also metrizable once we set
\begin{align*}
    d_\Omega(\omega,\eta):=\sum_{i\in\mathbb{Z}^d}2^{-\|i\|_\infty}\frac{d_M(\omega_i,\eta_i)}{1+d_M(\omega_i,\eta_i)}\,,
\end{align*}
where $\|i\|_\infty:=\sup\limits_{k\in\{1,\ldots,d\}}|i_k|$.

If $\Lambda\subset\mathbb{Z}^d$ we set $\Omega_\Lambda:=M^\Lambda$ and denote by \begin{align*}
	\pi_\Lambda\colon\Omega\to\Omega_\Lambda
	\qquad \pi_\Lambda(\omega):=\omega_\Lambda:=\omega|_\Lambda\,,
\end{align*}
the projection over $\Omega_\Lambda$.
With a standard slight abuse of notation we will set, given $\omega_\Lambda\in\Omega_\Lambda$ and $\omega_{\Lambda^c}\in\Omega_{\Lambda^c}$, $\Lambda\subset\mathbb{Z}^d$, $\omega_\Lambda\omega_{\Lambda^c}:=\omega_\Lambda\otimes\omega_{\Lambda^c}\in\Omega$.
For a given $\Lambda\subset\mathbb{Z}^d$ we will write $\Lambda\Subset\mathbb{Z}^d$ if $|\Lambda|<+\infty$: Notice that in this latter case $C(\Omega_\Lambda)\simeq C(M^{|\Lambda|})$ where $M^{|\Lambda|}$ is again a connected compact symplectic manifold with symplectic form $\varsigma_\Lambda:=\oplus_{i\in\Lambda}\varsigma$.

The algebra of observables over the configuration space $\Omega$ is identified with the commutative $C^*$-algebra $C(\Omega)$ equipped with the supremum norm $\|f\|_\infty:=\sup_{\omega\in\Omega}|f(\omega)|$.
This space is best described in terms of (continuous) local functions.
In particular, $f\in C(\Omega)$ is called \textbf{(continuous) local function}, denoted $f\in C_{\textsc{loc}}(\Omega)$, if there exists $\Lambda\Subset\mathbb{Z}^d$ and $f_\Lambda\in C(\Omega_\Lambda)\simeq C(M^{|\Lambda|})$ such that $f(\omega)=f_\Lambda(\omega_\Lambda)$ for all $\omega\in\Omega$.
Then, for all $f\in C(\Omega)$ there exists a sequence $(f_n)_{n\in\mathbb{N}}$ of local functions such that $\lim_{n\to\infty}\|f-f_n\|_\infty=0$ \cite[Lem. 6.21]{Friedly_Velenik_2017} (\textit{e.g.} one considers $\eta\in\Omega$ and a sequence $(\Lambda_n)_{n\in\mathbb{N}}$ of increasing subsets $\Lambda_n\subset\Lambda_{n+1}\Subset\mathbb{Z}^d$ such that $\cup_n\Lambda_n=\Omega$ and set $f_n(\omega):=f(\omega_{\Lambda_n}\eta_{\Lambda_n^c})$ for all $n\in\mathbb{N}$).
For this reason, $C(\Omega)$ is usually referred to as the $C^*$-algebra of \textbf{(continuous) quasi-local functions}.

Per definition, a \textbf{state} on $C(\Omega)$ is a normalized linear, positive functional $\varphi\colon C(\Omega)\to\mathbb{C}$.
By the Riesz-Markov-Kakutani theorem every state $\varphi$ on $C(\Omega)$ is completely described by a Radon probability measure $\mu\in\mathcal{P}(\Omega,\mathcal{F})$ on $(\Omega,\mathcal{F})$ ---here $\mathcal{F}$ denotes the Borel $\sigma$-algebra over $\Omega$.
In particular $\varphi\equiv\varphi_\mu$ and
\begin{align}
    \varphi_\mu(f)=\int_\Omega f\mathrm{d}\mu
    \qquad\forall f\in C(\Omega)\,.
\end{align}

For later convenience we recall the definition of the $\sigma$-algebra $\mathcal{F}_\Lambda$ over $\Omega$ of events occurring in $\Lambda\subset\mathbb{Z}^d$:
\begin{align*}
	\mathcal{F}_\Lambda:=
	\sigma\bigg(\bigcup_{\substack{\Lambda'\Subset\mathbb{Z}^d\\\Lambda'\subset\Lambda}}\{\pi_{\Lambda'}^{-1}(E)\,|\,E\in\mathcal{B}_{\Lambda'}\}\bigg)
	\,,
\end{align*}
where $\sigma(A)$ is the $\sigma$-algebra generated by the collection $A$ while, if $\Lambda\Subset\mathbb{Z}^d$, $\mathcal{B}_\Lambda$ denotes the Borel $\sigma$-algebra over $\otimes_{i\in\Lambda} M\simeq M^{|\Lambda|}$.
Notice that
\begin{align*}
	\mathcal{F}_\Lambda
	=\{\pi_{\Lambda}^{-1}(E)\,|\,E\in\mathcal{B}_{\Lambda}\}
	\qquad\forall \Lambda\Subset\mathbb{Z}^d\,,
\end{align*}
while $\mathcal{F}_{\mathbb{Z}^d}$ coincides with the Borel $\sigma$-algebra $\mathcal{F}$ over $\Omega$.
A function $f\colon\Omega\to\mathbb{C}$ is $\mathcal{F}_\Lambda$-measurable if and only if $f$ is $\Lambda\textbf{-local}$, namely there exists a measurable function $f_\Lambda\colon\Omega_\Lambda\to\mathbb{C}$ such that $f(\omega)=f_\Lambda(\omega_\Lambda)$ for all $\omega\in\Omega$ \cite[Lem. 6.3]{Friedly_Velenik_2017}.

Among all possible states $\varphi_\mu$ on $C(\Omega)$ we will be interested in the infinite volume Gibbs states.
The definition of the latter requires to introduce a \textbf{potential}, namely a collection $\Phi=\{\Phi_\Lambda\}_{\Lambda\Subset\mathbb{Z}^d}$, where $\Phi_\Lambda\in C(\Omega)$ is $\mathcal{F}_\Lambda$-measurable for all $\Lambda\Subset\mathbb{Z}^d$.
In what follows we will assume the following technical assumptions on $\Phi$:
\begin{enumerate}[(I)]
	\item\label{Item: Phi is C1}
	For all $\Lambda\Subset\mathbb{Z}^d$, $\Phi_\Lambda\in C^1(\Omega_\Lambda)\simeq C^1(M^{|\Lambda|})$.
	\item\label{Item: Phi is C1 summable}
	For all $i\in\mathbb{Z}^d$
	\begin{align}\label{Eq: Phi is C1 summable}
		\sum_{\substack{\Lambda\Subset\mathbb{Z}^d\\i\in\Lambda}}\|\Phi_\Lambda\|_{C^1(\Omega_\Lambda)}<+\infty\,,
	\end{align}
	where $\|\Phi_\Lambda\|_{C^1(M_\Lambda)}$ denotes the supremum of $\Phi_\Lambda$ and $\mathrm{d}\Phi_\Lambda$ over $\Omega_\Lambda$, thus it coincides with the norm of the $C^*$-algebra $C^1(M^{|\Lambda|})$.
\end{enumerate}
Within these assumptions we can define
\begin{align}\label{Eq: local Hamiltonian}
	H_\Lambda^\Phi
	:=\sum_{\substack{\Lambda'\Subset\mathbb{Z}^d\\\Lambda'\cap\Lambda\neq\emptyset}}\Phi_{\Lambda'}\,,
\end{align}
which is a well-defined element in $C(\Omega)$ on account of assumption \eqref{Item: Phi is C1 summable}.
For a fixed $\eta\in\Omega$ one then introduces the probability measure $\mu^\Phi_\Lambda(\;|\eta)$ on $(\Omega,\mathcal{F})$ defined by
\begin{align}\label{Eq: Gibbs measure on Lambda with eta boundary conditions}
	\mu^\Phi_\Lambda(A|\eta)
	:=\frac{1}{Z^\Phi_{\eta,\Lambda}}\int_{\Omega_\Lambda} 1_A(\omega_\Lambda\eta_{\Lambda^c})
	e^{-H^\Phi_\Lambda(\omega_\Lambda\eta_{\Lambda^c})}\mathrm{d}\mu_{M^{|\Lambda|}}(\omega_\Lambda)
	\quad\forall A\in\mathcal{F}\,,
\end{align}
where $\mu_{M^{|\Lambda|}}:=\otimes_{i\in\Lambda}\mu_M$ denotes the (normalized) Liouville volume form over $M^{|\Lambda|}$ while $Z^\Phi_{\eta,\Lambda}$ is a normalization constant ---we have implicitly identified $\Omega_\Lambda\simeq M^{|\Lambda|}$.
We will refer to $\mu^\Phi_\Lambda(\;|\eta)$ as the \textbf{$\Lambda$-Gibbs measure} associated with the potential $\Phi$ and with boundary condition $\eta$ (at fixed inverse temperature $\beta=1$).
We will denote by $\varphi^\Phi_\Lambda(\;|\eta)$ the associated state on $C(\Omega)$ which will be referred to as the $\Lambda$-\textbf{Gibbs state} associated with $\Phi$ and $\eta$.

\begin{remark}\label{Rmk: Phi C1 assumption}
	The assumption on the potential $\Phi$ are slightly stronger than those usually imposed \cite[\S 6]{Friedly_Velenik_2017}.
	As a matter of fact, one usually requires
	\begin{align*}
		\sum_{\substack{\Lambda\Subset\mathbb{Z}^d\\i\in\Lambda}}
		\|\Phi_\Lambda\|_{C(\Omega_\Lambda)}<+\infty\,,
	\end{align*}
	which ensures the well-definiteness of $H^\Phi_\Lambda$ as per Equation \eqref{Eq: local Hamiltonian}.
	The stronger assumptions \eqref{Item: Phi is C1}-\eqref{Item: Phi is C1 summable} are necessary for discussing the Poisson bracket on $\Omega$ ---\textit{cf.} Section \ref{Sec: Classical KMS condition on lattice systems}.
\end{remark}

The collection $\mu^\Phi:=\{\mu^\Phi_\Lambda\}_{\Lambda\Subset\mathbb{Z}^d}$ is called \textbf{Gibbs specification} and enjoys the following properties:
\begin{enumerate}[(a)]
	\item\label{Item: probability measure for fixed eta}
	For all $\Lambda\Subset\mathbb{Z}^d$ and $\eta\in\Omega$, $\mu^\Phi_\Lambda(\;|\eta)$ is a probability measure over $(\Omega,\mathcal{F})$;
	
	\item\label{Item: measurability for fixed measurable set}
	For all $\Lambda\Subset\mathbb{Z}^d$ and $A\in\mathcal{F}$, the function $\Omega\ni\eta\mapsto\mu^\Phi_\Lambda(A|\eta)$ is $\mathcal{F}_{\Lambda^c}$-measurable.
	
	\item[]
	Properties \eqref{Item: probability measure for fixed eta}-\eqref{Item: measurability for fixed measurable set} can be summarized by saying that $\mu^\Phi_\Lambda$ is a \textbf{probability kernel} from $\mathcal{F}_{\Lambda^c}$ to $\mathcal{F}$ \cite[\S 5]{Kallenberg_2021} (also known as transition measures \cite[Def. 10.7.1]{Bogachev_2007}).

	\item\label{Item: Gibbs kernel is proper}
	For all $\Lambda\Subset\mathbb{Z}^d$ and $A\in\mathcal{F}_{\Lambda^c}$, it holds $\mu^\Phi_\Lambda(A|\eta)=1_A(\eta)$.
	Probability kernels abiding by this assumption are called \textbf{proper}.
	
	\item\label{Item: Gibbs compatibility condition}
	The family of probability kernels $\mu^\Phi$ is \textbf{compatible}, namely, for all $\Lambda_1\subset\Lambda_2\Subset\mathbb{Z}^d$ it holds $\mu^\Phi_{\Lambda_2}\mu^\Phi_{\Lambda_1}=\mu^\Phi_{\Lambda_2}$, where
	\begin{align*}
		\mu^\Phi_{\Lambda_2}\mu^\Phi_{\Lambda_1}\colon
		\mathcal{F}\times\Omega
		\ni(A,\eta)\mapsto
		\int_\Omega\mu^\Phi_{\Lambda_1}(A|\omega)\mathrm{d}\mu^\Phi_{\Lambda_2}(\omega|\eta)\,.
	\end{align*}
\end{enumerate}

For any probability measure $\mu\in\mathcal{P}(\Omega,\mathcal{F})$ and $\Lambda\Subset\mathbb{Z}^d$ one may introduce a new probability measure $\mu\mu^\Phi_\Lambda\in\mathcal{P}(\Omega,\mathcal{F})$ defined by
\begin{align}\label{Eq: composition of probability measure and Gibbs kernel}
	\mu\mu^\Phi_\Lambda(A)
	:=\int_\Omega\mu^\Phi_\Lambda(A|\eta)\mathrm{d}\mu(\eta)\,.
\end{align}
A state $\varphi_\mu$ is called \textbf{infinite volume Gibbs state}, denoted $\varphi_\mu\in\mathcal{G}(\Phi)$, if the associated Radon probability measure $\mu$ fulfils
\begin{align}\label{Eq: oo-volume Gibbs state property}
	\mu\mu^\Phi_\Lambda=\mu\qquad\forall\Lambda\Subset\mathbb{Z}^d\,.
\end{align}
Any $\varphi_\mu\in\mathcal{G}(\Phi)$ is interpreted as a state in thermal equilibrium (at fixed inverse temperature $\beta=1$) with respect to the formal Hamiltonian $H_{\mathbb{Z}^d}:=\sum_{\Lambda\Subset\mathbb{Z}^d}\Phi_\Lambda$.
This condition allows to characterize thermal equilibrium states on the whole configuration space $\Omega$, without resorting to a finite configuration space $\Omega_\Lambda$, $\Lambda\Subset\mathbb{Z}^d$.
This approach provides a rather neat definition of phase transition: A \textbf{phase transition} occurs whenever $\mathcal{G}(\Phi)$ is not a singleton.
In the physical jargon this is equivalent to the existence of different equilibrium states in the infinite volume system.

The compatibility condition \eqref{Eq: oo-volume Gibbs state property} with the Gibbs specification reflect the property that locally, \textit{i.e.} on $\Omega_\Lambda$ for $\Lambda\Subset\mathbb{Z}^d$ with fixed boundary conditions on $\Lambda^c$, all equilibrium states coincide with the Gibbs measure \eqref{Eq: Gibbs measure on Lambda with eta boundary conditions}.
In physical terms, there are no phase transitions in a finite system.

As we shall see in the next section, in the current situation the symplectic structure on $M$ allows for a different characterization of thermal equilibrium states by means of the (classical) KMS condition \cite{Gallavotti_Pulvirenti_1976,Gallavotti_Verboven_1975}.
The goal of this paper is to prove that these two characterizations of thermal states coincide.

\begin{remark}\label{Rmk: existence of regular conditional probability}
    Condition \eqref{Eq: oo-volume Gibbs state property} can be interpreted by saying that, for all $\Lambda\Subset\mathbb{Z}^d$, $\mu^\Phi_\Lambda(\;|\eta)$ is a regular conditional measure for $\mu$ with respect to $\mathcal{F}_{\Lambda^c}$ \cite[Def. 10.4.1]{Bogachev_2007}.
    For later convenience we stress that any probability measure $\mu\in\mathcal{P}(\Omega,\mathcal{F})$ arising from a state $\varphi_\mu$ on $C(\Omega)$ has a \textbf{regular conditional measure} $\mu_\Lambda$ with respect to $\mathcal{F}_{\Lambda^c}$ for all $\Lambda\Subset\mathbb{Z}^d$.
    This follows from \cite[Cor. 10.4.7]{Bogachev_2007}, together with the observation that:
    (a) since $M$ is metrizable, $\Omega$ is second-countable (and also a \textbf{Polish space} because it is metrizable), therefore, $\mathcal{F}$ is countably generated;
    (b) since $\Omega$ is compact, the Riesz-Markov-Kakutani theorem ensures that the probability measure $\mu$ associated with any state $\varphi_\mu$ is Radon, thus it has an approximating compact class (made by compact subsets of $\Omega$).
    Thus, for all $\Lambda\Subset\mathbb{Z}^d$ there exists
    \begin{align*}
		\mu_\Lambda\colon\mathcal{F}\times\Omega
		\ni(A,\eta)\mapsto\mu_\Lambda(A|\eta)\in[0,1]\,,
    \end{align*}
    such that $\mu_\Lambda(\;|\eta)\in\mathcal{P}(\Omega,\mathcal{F})$ for $\mu$-almost all $\eta\in\Omega$ and $\mu_\Lambda(A|\;)$ is $\mathcal{F}_{\Lambda^c}$-measurable, moreover,
    \begin{align}\label{Eq: regular conditional measure condition}
        \mu(A\cap B)=\int_B \mu_\Lambda(A|\eta)\mathrm{d}\mu(\eta)
		\qquad
		\forall A\in\mathcal{F}\,,\,
		\forall B\in\mathcal{F}_{\Lambda^c}\,.
    \end{align}
    Thus $\varphi_\mu\in\mathcal{G}(\Phi)$ if and only if $\mu_\Lambda(\;|\eta)=\mu^\Phi_\Lambda(\;|\eta)$ for $\mu$-almost all $\eta\in\Omega$ and for all $\Lambda\Subset\mathbb{Z}^d$.
    It is worth observing that, since $\mathcal{F}$ is countably generated, $\mu_\Lambda(\;|\eta)$ is unique for all $\eta\in\Omega$ up to a $\mu$-null set \cite[Lem. 10.4.3]{Bogachev_2007}.
    Moreover, since $\mu_\Lambda(\;|\eta)$ is $\mu$-integrable, a density argument shows that, for all $f\in C(\Omega)$,
    \begin{align*}
        \varphi_\mu(f)=\int_\Omega\varphi_\Lambda(f|\eta)\mathrm{d}\mu(\eta)\,,
    \end{align*}
    where $\varphi_\Lambda(\;|\eta)$ denotes the state on $C(\Omega)$ associated with $\mu_\Lambda(\;|\eta)$.
    In particular, if $f,g\in C(\Omega)$ and $g$ is $\mathcal{F}_{\Lambda^c}$-measurable, $\Lambda\Subset\mathbb{Z}_+^d$, then
    \begin{align}\label{Eq: integrating a Lambdac-measurable function}
        \varphi_\mu(fg)=\int_\Omega\varphi_\Lambda(f|\eta)g(\eta)\mathrm{d}\mu(\eta)\,.
    \end{align}
	Finally, it is worth observing that, as a consequence of \cite[Cor. 10.4.10]{Bogachev_2007} $\varphi_\Lambda(\;|\eta)$ is concentrated on $\pi_{\Lambda^c}^{-1}(\eta_{\Lambda^c})$ for all $\Lambda\Subset\mathbb{Z}^d$ and $\mu$-almost all $\eta\in\Omega$.
	Notice that $\pi_{\Lambda^c}^{-1}(\eta_{\Lambda^c})\in\mathcal{F}_{\Lambda^c}$ because $\pi_{\Lambda^c}^{-1}(\eta_{\Lambda^c})=\bigcap_{n\in\mathbb{N}}\pi_{\Lambda_n\setminus\Lambda}^{-1}(\eta_{\Lambda_n\setminus\Lambda^c})$ where $(\Lambda_n)_{n\in\mathbb{N}}$ is an increasing sequence of finite subsets $\Lambda_n\Subset\mathbb{Z}^d$ such that $\bigcup_{n\in\mathbb{N}}\Lambda_n=\mathbb{Z}^d$.
\end{remark}

\section{Classical KMS condition on lattice systems}
\label{Sec: Classical KMS condition on lattice systems}

This section is devoted to describe the classical KMS condition on the lattice system $\Omega=M^{\mathbb{Z}^d}$.
This condition is used to select a particular class of states $\varphi_\mu$ by means of the Poisson structure carried by $M$ \cite{Aizenman_Gallavotti_Goldstein_Lebowitz_1976,Aizenman_Goldstein_Gruber_Lebowitz_Martin_1977,Gallavotti_Pulvirenti_1976,Gallavotti_Verboven_1975}.
This is parallel to the notion of infinite volume Gibbs states \eqref{Eq: oo-volume Gibbs state property} which is instead based on the $\Lambda$-local Gibbs measure \eqref{Eq: Gibbs measure on Lambda with eta boundary conditions}.

To set the stage we recall that, for all $\Lambda\Subset\mathbb{Z}^d$, $\Omega_\Lambda\simeq M^{|\Lambda|}$ is a compact, connected symplectic manifold with $\varsigma_\Lambda:=\oplus_{i\in\Lambda}\varsigma$.
For $f,g\in C^\infty(\Omega_\Lambda)$ we denote by $\{f,g\}_\Lambda$ the associated Poisson bracket.

\begin{remark}\label{Rmk: functional invariant under Hamiltonian flows}
	For later convenience we recall the following property of the state $\varphi_{\Omega_\Lambda}$ on $C(\Omega_\Lambda)\simeq C(M^{|\Lambda|})$ induced by the normalized Liouville measure $\mu_{M^{|\Lambda|}}$ on $M^{|\Lambda|}$.
	Actually, $\varphi_{\Omega_\Lambda}$ is the unique state invariant under all Hamiltonian vector fields.
	In fact, if $\psi_\Lambda$ is a state on $C(\Omega_\Lambda)$ such that
	\begin{align*}
		\psi_\Lambda(\{f,g\}_\Lambda)=0
		\qquad
		\forall f,g\in C^\infty(\Omega_\Lambda)\,,
	\end{align*}
	then $\psi_\Lambda=\varphi_{\Omega_\Lambda}$ \textit{i.e.}
	\begin{align*}
		\psi_\Lambda(f)=\int_{\Omega_\Lambda}f(\omega_\Lambda)\mathrm{d}\mu_{M^{|\Lambda|}}(\omega_\Lambda)\,.
	\end{align*}
	The proof of this result can be found in \cite[Lem. 1]{Aizenman_Goldstein_Gruber_Lebowitz_Martin_1977} as well as in \cite[Cor. 2.7]{Bordermann_Romer_Waldmann_1998}.
	It is based on a localization process in coordinated charts equipped with Darboux coordinates ---so that one can write $\partial_a f=\{f,x^a\}_\Lambda$--- together with the observation that the Lebesgue measure on $\mathbb{R}^d$ is the unique (up to multiplicative constants) measure vanishing on derivatives of compactly supported functions \cite[Lem. 2.6]{Bordermann_Romer_Waldmann_1998}.
\end{remark}

The KMS condition requires to endow $C(\Omega)$ with a Poisson bracket.
We recall that a \textbf{Poisson bracket} $\{\;,\;\}$ on $C(\Omega)$ is a skew-symmetric bilinear map $\{\;,\;\}\colon D\times D\to D$ defined on a dense sub-algebra $D$ of $C(\Omega)$ such that $\{\;,f\}$ is a derivation on $D$ for all $f\in D$ and fulfilling the \textbf{Jacobi identity}
\begin{align*}
		\{\{f,g\},h\}
		=\{\{f,h\},g\}
		+\{f,\{g,h\}\}
		\qquad\forall f,g,h\in D\,.
\end{align*}
In what follows we will consider $D=C^\infty_{\textsc{loc}}(\Omega)$ where $C^\infty_{\textsc{loc}}(\Omega)$ is the algebra of smooth local functions: $f\in C^\infty_{\textsc{loc}}(\Omega$) if there exists $\Lambda\Subset\mathbb{Z}^d$ and $f_\Lambda\in C^\infty(\Omega_\Lambda)$ such that $f(\omega)=f_\Lambda(\omega_\Lambda)$.

To define a Poisson bracket on $C(\Omega)$ we notice that, $\{f,g\}_\Lambda$ makes sense also for $f,g\in C^\infty(\Omega_{\Lambda'})$ for all $\Lambda\subset\Lambda'\Subset\mathbb{Z}^d$.
As a matter of fact, if $\Lambda_1\subset\Lambda_2\Subset\mathbb{Z}^d$, denoting $\pi^2_1\colon\Omega_{\Lambda_2}\to\Omega_{\Lambda_1}$ the projection $\omega_{\Lambda_2}=\omega_{\Lambda_1}\omega_{\Lambda_2\setminus\Lambda_1}\mapsto\omega_{\Lambda_1}$, we find
\begin{align}\label{Eq: compatibility of Poisson brackets}
	\{(\pi^2_1)^*f,g\}_{\Lambda_2}
	=\{(\pi^2_1)^*f,g\}_{\Lambda_1}\,,
\end{align}
for all $f\in C^\infty(\Omega_{\Lambda_1})$ and $g\in C^\infty(\Omega_{\Lambda_2})$ where $(\pi^2_1)^*\colon C(\Omega_{\Lambda_1})\to C(\Omega_{\Lambda_2})$ ---in what follows we will not write $\pi^2_1$ since its use will be clear from the context.
The Poisson bracket $\{\;,\;\}$ on $C(\Omega)$ is defined by
\begin{align}\label{Eq: Poisson bracket on Omega}
	\{f,g\}:=\{f,g\}_{\Lambda}\,,
\end{align}
where $f,g\in C^\infty_{\textsc{loc}}(\Omega)$ while $\Lambda\Subset\mathbb{Z}^d$ is such that $f,g\in C^\infty(\Omega_{\Lambda})$.
On account of Equation \eqref{Eq: compatibility of Poisson brackets} the value of $\{f,g\}_{\Lambda}$ does not change if we enlarge $\Lambda$.
By direct inspection $\{\;,\;\}$ defines a Poisson bracket on $C(\Omega)$.

The next ingredient for stating the KMS condition is the choice of a vector field $X$ on $C(\Omega)$ which plays the role of the infinitesimal generator of the dynamics on $\Omega$.
Notice that, despite this interpretation, the (classical) KMS condition does not rely on the existence of a dynamics integrating $X$.
We now consider the vector field
\begin{align}\label{Eq: local Hamiltonian vector field}
	X^\Phi_{\Lambda}\colon C^\infty_{\textsc{loc}}(\Omega)\to C(\Omega)
        \qquad
	X^\Phi_\Lambda(f)
	:=\{f,H^\Phi_\Lambda\}
        =\sum_{\substack{\Lambda'\Subset\mathbb{Z}^d\\\Lambda'\cap\Lambda\neq\emptyset}}
	\{f,\Phi_{\Lambda'}\}_{\Lambda'}\,,
\end{align}
where the series converges on account of assumption \eqref{Item: Phi is C1 summable}.
Notice that, $\{f,H^\Phi_\Lambda\}$ is only continuous on $\Omega$, moreover, it is not a local function.

By direct inspection we also have, for all $\Lambda''\subseteq\Lambda\Subset\mathbb{Z}^d$ and $f\in C^\infty(\Omega_{\Lambda''})$,
\begin{align}\label{Eq: compatibility of local Hamiltonian vector field}
	\{f,H^\Phi_\Lambda\}
	=\sum_{\substack{\Lambda'\Subset\mathbb{Z}^d\\\Lambda'\cap\Lambda\neq\emptyset}}
	\{f,\Phi_{\Lambda'}\}_{\Lambda'\cap\Lambda''}
	=\sum_{\substack{\Lambda'\Subset\mathbb{Z}^d\\\Lambda'\cap\Lambda''\neq\emptyset}}
	\{f,\Phi_{\Lambda'}\}_{\Lambda'\cap\Lambda''}
	=\{f,H^\Phi_{\Lambda''}\}\,.
\end{align}
Equation \eqref{Eq: compatibility of local Hamiltonian vector field} implies that, if $f$ is $\Lambda''$-local and $\Lambda''\subseteq\Lambda\Subset\mathbb{Z}^d$, then $X^\Phi_\Lambda(f)$ does not change if we enlarge $\Lambda$.
This allows to introduce a global vector field $X^\Phi$ defined by
\begin{align}\label{Eq: global Hamiltonian vector field}
    X^\Phi\colon C^\infty_{\textsc{loc}}(\Omega)\to C(\Omega)
    \qquad
    X^\Phi(f)=\{f,H^\Phi_\Lambda\}
    \qquad
    \forall f\in C^\infty(\Omega_\Lambda)\,.
\end{align}
Loosely speaking, one may think of $X^\Phi$ as the Hamiltonian vector field associated to the formal Hamiltonian $H=\sum_{\Lambda\Subset\mathbb{Z}^d}\Phi_\Lambda$.

We can finally introduce the KMS condition associated with the vector field $X^\Phi$ and the Poisson bracket $\{\;,\;\}$.
The latter condition select a particular subclass of states $\varphi_\mu$ on $C(\Omega)$ with a constraint on the expectation value of Poisson brackets.

In more details, a state $\varphi_\mu$ on $C(\Omega)$ is called \textbf{$X^\Phi$-KMS state} if
\begin{align}\label{Eq: KMS condition}
    \varphi_\mu(\{f,g\})=\varphi_\mu(gX^\Phi(f))
    \qquad
    \forall f,g\in C^\infty_{\textsc{loc}}(\Omega)\,.
\end{align}
We denote by $\mathcal{K}(\Phi)$ the convex set of KMS states.

\section{Proof of the classical DLR-KMS correspondence}
\label{Sec: Proof of the classical DLR-KMS correspondence}

This section is devoted to the proof of Theorem \ref{Thm: equivalence of oo-volume Gibbs states and KMS states} which asserts the equivalence between the DLR condition \eqref{Eq: oo-volume Gibbs state property} and the KMS condition \eqref{Eq: KMS condition}.
The argument profits of the one presented in \cite[Lem. 2]{Aizenman_Goldstein_Gruber_Lebowitz_Martin_1977}.

To lighten the presentation, the proof of Theorem \ref{Thm: equivalence of oo-volume Gibbs states and KMS states} is divided in two propositions.

\begin{proposition}\label{Prop: oo-volume Gibbs states are KMS states}
    It holds $\mathcal{G}(\Phi)\subseteq\mathcal{K}(\Phi)$, namely every infinite volume Gibbs state is also a $X^\Phi$-KMS state. 
\end{proposition}
\begin{proof}
    Let $\varphi_\mu\in\mathcal{G}(\Phi)$, $\Lambda\Subset\mathbb{Z}^d$ and $f,g\in C^\infty(\Omega_\Lambda)$.
    Equation \eqref{Eq: oo-volume Gibbs state property} implies
    \begin{align*}
        \varphi_\mu(\{f,g\})
        =\int\varphi^\Phi_\Lambda(\{f,g\}_\Lambda|\eta)\mathrm{d}\mu(\eta)\,.
    \end{align*}
    By direct inspection we also have
    \begin{align*}
        \varphi^\Phi_\Lambda(\{f,g\}_\Lambda|\eta)
        &:=\frac{1}{Z^\Phi_{\eta,\Lambda}}\int_{\Omega_\Lambda}
        \{f,g\}_\Lambda(\omega_\Lambda\eta_{\Lambda^c})
	e^{-H^\Phi_\Lambda(\omega_\Lambda\eta_{\Lambda^c})}\mathrm{d}\mu_{M^{|\Lambda|}}(\omega_\Lambda)
        \\
        &=\frac{1}{Z^\Phi_{\eta,\Lambda}}\int_{\Omega_\Lambda}
        \{f,g
	e^{-H^\Phi_\Lambda}\}_\Lambda(\omega_\Lambda\eta_{\Lambda^c})\mathrm{d}\mu_{M^{|\Lambda|}}(\omega_\Lambda)
        \\
        &+\frac{1}{Z^\Phi_{\eta,\Lambda}}\int_{\Omega_\Lambda}
        g\{f,H_\Lambda^\Phi \}_\Lambda(\omega_\Lambda\eta_{\Lambda^c})
	e^{-H^\Phi_\Lambda(\omega_\Lambda\eta_{\Lambda^c})}\mathrm{d}\mu_{M^{|\Lambda|}}(\omega_\Lambda)
        \\
        &=\varphi^\Phi_\Lambda(g\{f,H^\Phi_\Lambda\}_\Lambda|\eta)
        =\varphi^\Phi_\Lambda(gX^\Phi(f)|\eta)\,,
    \end{align*}
    where in the last line we used that $\varphi_{\Omega_\Lambda}(\{h_1,h_2\}_\Lambda)=0$ for all $h_1,h_2\in C^1(\Omega_\Lambda)$.
    Overall we have
    \begin{align*}
        \varphi_\mu(\{f,g\})
        =\int_\Omega\varphi^\Phi_\Lambda(gX^\Phi(f)|\eta)\mathrm{d}\mu(\eta)
        =\varphi_\mu(gX^\Phi(f))\,,
    \end{align*}
    therefore, $\varphi_\mu\in\mathcal{K}(\Phi)$.
\end{proof}

\begin{proposition}\label{Prop: KMS states are oo-volume Gibbs states}
	It holds $\mathcal{K}(\Phi)\subseteq\mathcal{G}(\Phi)$, namely every $X^\Phi$-KMS state is an infinite volume Gibbs state. 
\end{proposition}
\begin{proof}
    Using Equation \eqref{Eq: regular conditional measure condition} it is enough to prove that $\varphi_\Lambda(\;|\eta)=\varphi^\Phi_\Lambda(\;|\eta)$ for $\mu$-almost all $\eta\in\Omega$ and all $\Lambda\Subset\mathbb{Z}^d$.
    We recall that for each $f\in C(\Omega)$, the function $\eta\mapsto\varphi_\Lambda(f|\eta)$ is $\mathcal{F}_{\Lambda^c}$-measurable, thus, it only depends on $\eta_{\Lambda^c}$.
    Moreover, $\varphi_\Lambda(\;|\eta)$ is concentrated on $\pi_{\Lambda^c}^{-1}(\eta_{\Lambda^c})$ $\mu$-almost all $\eta\in\Omega$.
	
	Let $\Lambda\Subset\mathbb{Z}^d$, $\Lambda'\Subset\Lambda^c$ and consider $f\in C^\infty(\Omega_\Lambda)$, $h\in C^\infty(\Omega_{\Lambda'})$, $g\in C^\infty_{\textsc{loc}}(\Omega)$.
	Then Equation \eqref{Eq: compatibility of Poisson brackets} implies that
	\begin{align*}
		h\{f,g\}=\{f,gh\}\,,
	\end{align*}
	since $\{f,h\}=0$.
	Moreover, Equation \eqref{Eq: integrating a Lambdac-measurable function} and the KMS condition \eqref{Eq: KMS condition} imply
	\begin{multline}\label{Eq: integrated Lambda KMS condition}
		\int_\Omega\varphi_\Lambda(\{f,g\}|\eta)h(\eta)\mathrm{d}\mu(\eta)
		=\varphi_\mu(h\{f,g\})
		=\varphi_\mu(\{f,gh\})
		\\
		=\varphi_\mu(ghX^\Phi(f))
		=\int_\Omega\varphi_\Lambda(g\{f,H^\Phi_\Lambda\}|\eta)h(\eta)\mathrm{d}\mu(\eta)\,.
	\end{multline}
	The arbitrariness of $h\in C^\infty(\Omega_{\Lambda'})\subset C(\Omega_{\Lambda^c})$ and of $\Lambda'\Subset\Lambda^c$ entails that, by an approximation argument,
    \begin{align*}
        \int_\Omega\big[\varphi_\Lambda(\{f,g\}_\Lambda|\eta)-\varphi_\Lambda(g\{f,H^\Phi_\Lambda\}_\Lambda|\eta)\big]h(\eta)\mathrm{d}\mu(\eta)
        =0
        \qquad\forall h\in C(\Omega_{\Lambda^c})\,.
    \end{align*}
	Since $\eta\mapsto\varphi_\Lambda(f|\eta)$ in $\mathcal{F}_{\Lambda^c}$-measurable, it follows that there exists $N_{f,g}\in\mathcal{F}_{\Lambda^c}\subset\mathcal{F}$ with $\mu(N_{f,g})=0$ such that
    \begin{align}\label{Eq: Lambda KMS condition}
		\varphi_\Lambda(\{f,g\}_\Lambda|\eta)
		=\varphi_\Lambda(g\{f,H^\Phi_\Lambda\}_\Lambda|\eta)
		\qquad\forall\eta\in\Omega\setminus N_{f,g}\,.
    \end{align}
    To proceed, we need to cope with the $f,g$-dependence of $N_{f,g}$.
    However, since $C(\Omega_\Lambda)$ is separable, we may choose $N_{f,g}$ independently on $f,g$: Indeed, it suffices to consider Equation \eqref{Eq: Lambda KMS condition} for $f,g$ on a countable dense set of $C^\infty(\Omega_\Lambda)$.
    This leads to countably many $\mu$-null sets whose union leads to the $\mu$-null set $N$ of interest.
	
	We now prove that condition \eqref{Eq: Lambda KMS condition} implies that $\varphi_\Lambda(\;|\eta)=\varphi^\Phi_\Lambda(\;|\eta)$ for all $\eta\in\Omega\setminus N$.
	For each $\eta\in\Omega\setminus N$ we consider the linear positive functional $\psi_\Lambda$ on $C(\Omega_\Lambda)$ defined by
	\begin{align*}
		\psi_\Lambda(f)
		:=\varphi_\Lambda(fe^{H^\Phi_\Lambda}|\eta)
		\qquad
		\forall f\in C(\Omega_\Lambda)\,.
	\end{align*}
	By direct inspection we have, for all $f,g\in C^\infty(\Omega_\Lambda)$,
	\begin{multline*}
		\psi_\Lambda(\{f,g\}_\Lambda)
		=\varphi_\Lambda(\{f,g\}_\Lambda e^{H^\Phi_\Lambda})
		\\
		=\varphi_\Lambda(\{f,ge^{H^\Phi_\Lambda}\}_\Lambda)
		-\varphi_\Lambda(g\{f,H^\Phi_\Lambda\}_\Lambda e^{H^\Phi_\Lambda})
		=0\,,
	\end{multline*}
	where we applied Equation \eqref{Eq: Lambda KMS condition}.
	
	Thus, $\psi_\Lambda$ is a linear, positive functional on $C(\Omega_\Lambda)\simeq C(M^{|\Lambda|})$ which is invariant under all Hamiltonian vector fields.
	Remark \ref{Rmk: functional invariant under Hamiltonian flows} implies that $\widehat{\psi}_\Lambda=\varphi_{\Omega_\Lambda}$ where $\widehat{\psi}_\Lambda(f)=\psi_\Lambda(f)/\psi_\Lambda(1)$.
	Finally, since $\varphi_\Lambda(\;|\eta)$ is concentrated on $\pi_{\Lambda^c}^{-1}(\eta_{\Lambda^c})$ we have
        \begin{align*}
	      \varphi_\Lambda(f|\eta)
            =\widehat{\psi}_\Lambda(f_\eta e^{-H^\Phi_{\Lambda,\eta}})
            =\varphi^\Phi_\Lambda(f|\eta)\,,
	\end{align*}
	where $f_\eta e^{-H^\Phi_{\Lambda,\eta}}\in C(\Omega_\Lambda)$ is defined by $(f_\eta e^{-H^\Phi_{\Lambda,\eta}})(\omega_\Lambda):=(fe^{-H^\Phi_\Lambda})(\omega_\Lambda\eta_{\Lambda^c})$.
\end{proof}

\paragraph{Acknowledgements.}
We are grateful to prof. Aernout van Enter for his suggestions and helpful discussions on this project. We also thank Klaas Landsman, Rodrigo Bissacot, Teun van Nuland, Y. Velenik and Stefan Waldmann for their fruitful comments.
C. J. F. van de Ven is supported by a postdoctoral fellowship granted by the Alexander von Humboldt Foundation (Germany).
N.D. acknowledges the support of the GNFM-INdAM Progetto Giovani \textit{Non-linear sigma models and the Lorentzian Wetterich equation}.

\paragraph{Data availability statement.}
Data sharing is not applicable to this article as no new data were created or analysed in this study.

\paragraph{Conflict of interest statement.}
The authors certify that they have no affiliations with or involvement in any
organization or entity with any financial interest or non-financial interest in
the subject matter discussed in this manuscript.


\begin{thebibliography}{99}
	\bibitem{Aizenman_Gallavotti_Goldstein_Lebowitz_1976}
	Aizenman M., Gallavotti G., Goldstein S., Lebowitz J.L.,
	\textit{Stability and equilibrium states of infinite classical systems},
	Commun.Math. Phys. 48, 1-14 
	\href{https://doi.org/10.1007/BF01609407}
	{(1976).}
	
	\bibitem{Aizenman_Goldstein_Gruber_Lebowitz_Martin_1977}
	Aizenman M., Goldstein S., Gruber C., Lebowitz J.L., Martin P.,
	\textit{On the equivalence between KMS-states and equilibrium states for classical systems},
	Commun. Math. Phys. 53, 209-220
	\href{https://doi.org/10.1007/BF01609847}
	{(1977).}
	
	\bibitem{Bayen_Flato_Fronsdal_Lichnerowicz_1978_I}
	Bayen F., Flato M., Fronsdal C., Lichnerowicz A., Sternheimer D.,
	\textit{Deformation Theory and Quantization. 1.Deformations of Symplectic Structures}, 
	Annals Phys. 111, 61 
	\href{https://doi.org/10.1016/0003-4916(78)90224-5}
	{(1978).}
	
	\bibitem{Bayen_Flato_Fronsdal_Lichnerowicz_1978_II}
	Bayen F., Flato M., Fronsdal C., Lichnerowicz A., Sternheimer D.,
	\textit{Deformation Theory and Quantization. 2. Physical Applications},
	Annals Phys. 111-151 
	\href{https://doi.org/10.1016/0003-4916(78)90225-7}
	{(1978).}
		
	\bibitem{Basart_Flato_Lichnerowicz_Sternheimer_1984}
	Basart H., Flato M., Lichnerowicz A., Sternheimer D.,
	\textit{Deformation theory applied to quantization and statistical mechanics},
	Lett Math Phys 8, 483-494
	\href{https://doi.org/10.1007/BF00400978}
	{(1984).}
	

	\bibitem{Berezin_1975}
	Berezin F.A.,
	\textit{General concept of quantization},
	Commun. Math. Phys. 40, 153-174 
	\href{https://doi.org/10.1007/BF01609397}
	{(1975).}


	\bibitem{Bogachev_2007}
	Bogachev V.I.,
	\textit{Measure Theory Vol. I-II},
	Springer Berlin, Heidelberg 
	\href{https://doi.org/10.1007/978-3-540-34514-5}
	{(2007)}.

	\bibitem{Bordermann_Romer_Waldmann_1998}
	Bordermann M., R\"{o}mer H., Waldmann S.,
	\textit{A Remark on Formal KMS States in Deformation Quantization},
	Letters in Mathematical Physics, Vol. 45, 49-61 
	\href{https://doi.org/10.1023/A:1007481019610}
	{(1998).}
	
	\bibitem{Bratteli_Robinson_1981_II}
	Bratteli O., Robinson D.W.,
	\textit{Operator Algebras and Quantum Statistical Mechanics. Vol. II: Equilibrium States, Models in Statistical Mechanics},
	Springer Berlin, Heidelberg
	\href{https://doi.org/10.1007/978-3-662-09089-3}
	{(1981).}


    
    
    \bibitem{Dobrushin_1968_1}
	Dobrushin R.L.,
	\textit{The description of a random field by means of conditional probabilities and conditions of its regularity},
	Theor. Prob. Appl. 13, 197-224
	\href{https://doi.org/10.1137/1113026}
	{(1968)}.
	
	\bibitem{Dobrushin_1968_2}
	Dobrushin R.L.,
	\textit{Gibbsian random fields for lattice systems with pairwise interactions},
	Funct. Anal. Appl. 2, 292-301
	\href{https://doi.org/10.1007/BF01075681}
	{(1968).}
	
	\bibitem{Dobrushin_1968_3}
	Dobrushin R.L.,
	\textit{The problem of uniqueness of a Gibbs random field and the problem of phase transition},
	Funct. Anal. Appl. 2, 302-312
	\href{https://doi.org/10.1007/BF01075682}
	{(1968).}
	
	\bibitem{Dobrushin_1970}
	Dobrushin R.L.,
	\textit{Prescribing a system of random variables by conditional distributions},
	Th. Prob. Appl. 17, 582-600
	\href{https://doi.org/10.1137/1115049}
	{(1970).}
	
	\bibitem{Drago_vandeVen_2022}
	Drago N., van de Ven  C.\,J.\,F.,
	\textit{Strict deformation quantization and local spin interactions},
	arXiv:2210.10697 [math-ph]
	\href{https://doi.org/10.48550/arXiv.2210.10697}
	{(2022).}
	
	\bibitem{Drago_Waldmann_2021}
	Drago N., Waldmann S.,
	\textit{Classical KMS Functionals and Phase Transitions in Poisson Geometry},
	arXiv:2107.04399 [math-ph]
	\href{https://doi.org/10.48550/arXiv.2107.04399}
	{(2021).}
	
	\bibitem{Friedly_Velenik_2017}
	Friedli S., Velenik Y.,
	\textit{Statistical mechanics of lattice systems: A concrete mathematical introduction},
	Cambridge University Press
	\href{https://doi.org/10.1017/9781316882603}
	{(2017).}
	
	\bibitem{Gallavotti_Pulvirenti_1976}
	Gallavotti G., Pulvirenti M.,
	\textit{Classical KMS Condition and Tomita-Takesaki Theory},
	Comm. Math. Phys. 46, Number 1, 1-9
	\href{https://doi.org/10.1007/BF01610495}
	{(1976).}

	\bibitem{Gallavotti_Verboven_1975}
	Gallavotti G., Verboven E.,
	\textit{On the classical KMS boundary condition},
	Nuov Cim B, 28, 274-286
	\href{https://doi.org/10.1007/BF02722820}
	{(1975).}

     \bibitem{Georgii_2011}
	Georgii, H.-O.
	\textit{Gibbs Measures and Phase Transitions},
	Berlin, New York: De Gruyter
	\href{https://doi.org/10.1515/9783110250329}
	{(2011).}
  
 
	\bibitem{Haag_Hugenholtz_Winnik_1967}
	Haag R., Hugenholtz N. M., Winnink M.,
	\textit{On the equilibrium states in quantum statistical mechanics},
	Comm. Math. Phys. 5:215-236
	\href{https://doi.org/10.1007/BF01646342}
	{(1967).}
	
	\bibitem{Kallenberg_2021}
	Kallenberg O.,
	\textit{Foundations of Modern Probability},
	Springer Nature Switzerland AG
	\href{https://doi.org/10.1007/978-3-030-61871-1}
	{(2021)}.
	    
	
	\bibitem{Lanford_Ruelle_1969}
	Lanford O.E., Ruelle. D.,
	\textit{Observables at infinity and states with short range correlations in statistical mechanics},
	Comm. Math. Phys., 13: 194-215,
	\href{https://doi.org/10.1007/BF01645487}
	{(1969).}
	
	\bibitem{Moretti_vandeVen_2022}
	Moretti V., van de Ven C.\,J.\,F.,
	\textit{The classical limit of Schr\"{o}dinger operators in the framework of Berezin quantization and spontaneous symmetry breaking as emergent phenomenon},
	International Journal of Geometric Methods in Modern Physics, Vol. 19, Iss. 01
	\href{https://doi.org/10.1142/S0219887822500037}{(2022)}.
	
	\bibitem{Murro-vandeVen_2022}
	Murro S., van de Ven  C.\,J.\,F.,
	\textit{Injective tensor products in strict deformation quantization}
	Math Phys Anal Geom 25, 2
	\href{https://doi.org/10.1007/s11040-021-09414-1}{(2022)}.

 
	
	\bibitem{Ruelle_1967}
	Ruelle D.,
	\textit{A variational formulation of equilibrium statistical mechanics and the Gibbs phase rule},
	Comm. Math. Phys. 5, 324-329
	\href{https://doi.org/10.1007/BF01646446}
	{(1967).}
	
	\bibitem{Ruelle_1969}
	Ruelle D.,
	\textit{Statistical mechanics: rigorous results},
	New York: Benjamin
	\href{https://doi.org/10.1142/4090}
	{(1969).}
	
	\bibitem{vandeVen_2022}
	van de Ven  C.\,J.\,F.,
	\textit{The classical limit and spontaneous symmetry breaking in algebraic quantum theory},
	Expositiones Mathematicae Vol. 40, Iss. 3, 543-571
	\href{https://doi.org/10.1016/j.exmath.2022.02.002}
	{(2022).}
	
	\bibitem{vandeVen_2022_b}
	van de Ven C.\,J.\,F.,
	\textit{KMS states and their classical limit},
	arXiv:2211.01755 [math-ph]
	\href{https://doi.org/10.48550/arXiv.2211.01755}
	{(2022).}
	
\end{thebibliography}
\end{document}